\documentclass[a4paper,11pt]{article}
\usepackage{latexsym}
\usepackage{amssymb}
\usepackage{amsfonts}
\usepackage{amsmath,amsthm}
\usepackage{mathrsfs}
\usepackage{bm}
\usepackage{enumerate}
\usepackage{stmaryrd}
\usepackage{xcolor}
\usepackage{commath}
\usepackage{array}
\usepackage{cases}
\usepackage{comment}

\makeatletter
\newenvironment{@abssec}[1]{%
    \if@twocolumn

    \section*{#1}%
    \else

      \vspace{.05in}\footnotesize
      \parindent .2in
 {\upshape\bfseries #1. }\ignorespaces
    \fi}

    {\if@twocolumn\else\par\vspace{.1in}\fi}

\newenvironment{keywords}{\begin{@abssec}{\keywordsname}}{\end{@abssec}}

\newenvironment{AMS}{\begin{@abssec}{\AMSname}}{\end{@abssec}}

\newcommand\keywordsname{Keywords}
\newcommand\AMSname{AMS subject classifications}
\newcommand\AMname{AMS subject classification}
\makeatother
\newcommand\restr[2]{{
\left.\kern-\nulldelimiterspace 
#1 
\vphantom{|} 
\right|_{#2} 
}}

\theoremstyle{plain}
\newtheorem{thm}{Theorem}[section]

\newtheorem{prop}[thm]{Proposition}

\theoremstyle{definition}
\newtheorem{dfn}[thm]{Definition}
\newtheorem{e.g.}{Example}
\theoremstyle{remark}
\newtheorem{rem}[thm]{Remark}
\theoremstyle{plain}

\newenvironment{pf}[1]
{{\noindent\textbf{Proof of #1.}}}{\hfill \qed}

\numberwithin{equation}{section}

\def\XXint#1#2#3{{\setbox0=\hbox{$#1{#2#3}{\int}$}
\vcenter{\hbox{$#2#3$}}\kern-.5\wd0}}

\newcommand{\D}{\mathcal{D}}

\newcommand{\R}{\mathbb R}

\newcommand{\N}{\mathbb N}

\newcommand{\ep}{\varepsilon}

\newcommand{\dl}{\delta}

\allowdisplaybreaks[1]

\title{\bf From Tsallis to KL: Convergence and Error Estimates for Tsallis-Regularized Optimal Transport}

\author{Takeshi Suguro and Toshiaki Yachimura}
\date{}

\begin{document}

\maketitle

\begin{abstract}
We study the Tsallis-to-Kullback--Leibler (KL) limit for entropy-regularized optimal transport with nonnegative bounded continuous costs. Fixing the regularization parameter $\varepsilon > 0$, we first derive an exact variational reformulation of Tsallis-regularized optimal transport in terms of the Tsallis information projection onto the set of couplings. The formula isolates an explicit correction term and thereby explains why, unlike in the KL case, the regularized transport problem and the corresponding information projection problem do not coincide exactly. We also establish existence and uniqueness for the Tsallis information projection. We then prove, with respect to the narrow topology, the $\Gamma$-convergence of the Tsallis-regularized functionals to the KL-regularized functional as $q\downarrow1$, together with narrow convergence of their unique minimizers. Finally, we obtain explicit error estimates of order $O(q-1)$ for both the regularized optimal transport values and the associated information projection values. These results quantify the passage from Tsallis regularization to the classical KL setting and clarify the relation between entropic regularization and information projection for $1 < q \leq 2$.
\end{abstract}



\pagestyle{plain}
\thispagestyle{plain}

\section{Introduction}\label{intro}

Entropy-regularized optimal transport is now widely used for comparing probability measures in image processing, computer graphics, and machine learning \cite{cuturi2013sinkhorn, patrini2020sinkhorn, solomon2015convolutional}. One of its main advantages is that the addition of a strictly convex entropy-type penalty leads to a strictly convex variational problem with a unique minimizer, while retaining essential geometric aspects of optimal transport. The most common choice of regularizer is the Kullback--Leibler (KL) divergence \cite{KL51}, for which the resulting problem can be solved efficiently by the Sinkhorn--Knopp algorithm \cite{sinkhorn1964relationship, SK1967}. The KL divergence also admits an information projection viewpoint, which clarifies its close connections to Schr\"odinger's probabilistic perspective and information geometry; see, e.g., \cite{AmariNagaoka2000}. This viewpoint traces back to Schr\"odinger's work on statistical physics \cite{schrodinger1931} and has since motivated a substantial body of work on Schr\"odinger bridges and their links to stochastic control; see, e.g., \cite{chen2021stochastic, Leonard2014, Nutz2022introduction} and references therein.

We briefly recall the KL-regularized formulation and its projection. 
Let $(X_1,\mu_1)$ and $(X_2,\mu_2)$ be Polish probability spaces and let
$c:X_1\times X_2\to\mathbb{R}$ be a nonnegative bounded continuous cost
function. Denote by $\Pi(\mu_1,\mu_2)$ the set of couplings with marginals
$(\mu_1, \mu_2)$, and set 
\begin{equation*}
m = \mu_1 \otimes \mu_2.
\end{equation*}
The KL-regularized optimal transport (OT) problem is defined by 
\begin{equation}\label{KL-reg OT}
  \mathrm{OT}_\ep(\mu_1, \mu_2)
  = \inf_{\pi\in\Pi(\mu_1,\mu_2)}\left\{\int_{X_1 \times X_2} c\,d\pi + \ep\,\mathrm{KL}(\pi,m)\right\}, 
\end{equation}
where $\mathrm{KL}$ denotes the Kullback--Leibler divergence \cite{KL51}, defined as follows: for a finite positive measure $\nu$ on a Polish space and a probability measure $\mu$, 
\begin{equation*}
\mathrm{KL}(\mu, \nu) 
= \left\{
  \begin{aligned}
  &\int \frac{d\mu}{d\nu} \log{\frac{d\mu}{d\nu}}\, d\nu && (\mu\ll\nu), \\
  &+ \infty && (\mathrm{otherwise}).
  \end{aligned}
  \right. 
\end{equation*}
Introducing the Gibbs kernel $K_\ep = e^{-c/\ep} m$, a direct computation shows that the KL-regularized OT \eqref{KL-reg OT} can equivalently be written as the following information projection problem:
\begin{equation}\label{KLproj pb}
  \mathrm{OT}_\ep(\mu_1, \mu_2) = \ep\inf_{\pi\in\Pi(\mu_1,\mu_2)} \mathrm{KL}(\pi, K_\ep).
\end{equation}
Thus, the KL-regularized OT \eqref{KL-reg OT} coincides with the information projection \eqref{KLproj pb}. This relationship is a special feature of the KL divergence. 

For a general $f$-divergence \cite{Csiszar1967, Csiszar1975}, one can also define an $f$-projection onto $\Pi(\mu_1,\mu_2)$, but its objective functional does not coincide with the corresponding regularized OT functional in general. To study this issue concretely, we consider the Tsallis divergence (Tsallis relative entropy) proposed by Tsallis \cite{T88}. It is a power-type extension of the KL divergence parameterized by $q > 1$ and reduces to the KL divergence as $q \downarrow 1$. This makes Tsallis regularization a natural test case for understanding and quantifying the gap between regularized OT and information projection.

For $q > 1$, we consider the Tsallis relative entropy
\begin{equation*}
  D_q(\mu,\nu)
  = \left\{\begin{aligned}
    &\frac{1}{q-1}\int \left[\left(\frac{d\mu}{d\nu}\right)^q-\frac{d\mu}{d\nu}\right]d\nu && (\mu\ll\nu), \\
    &+ \infty && (\mathrm{otherwise})
    \end{aligned}\right.
\end{equation*}
for a finite positive measure $\nu$ on a Polish space and a probability measure $\mu$, and define the following Tsallis-regularized OT functional
\begin{equation}\label{Tsallis-reg OT}
  \mathrm{OT}_{q,\ep}(\mu_1, \mu_2)
  = \inf_{\pi\in\Pi(\mu_1,\mu_2)}\left\{\int_{X_1 \times X_2} c\,d\pi + \ep D_q(\pi,m)\right\}.
\end{equation}

For $q \in \R\setminus\{1\}$, the $q$-logarithm and $q$-exponential are defined by
\begin{equation}\label{q-log and  q-exp}
\log_q{x} = \frac{x^{1 - q} - 1}{1 - q}\quad (x > 0), \qquad
  \exp_q{x} = \left(1 + (1-q)x\right)^{\frac{1}{1-q}}\quad (x \in I_q),
\end{equation}
where $I_q = \{x\in\R:1+(1-q)x>0\}$. For $q>1$, we use the continuous extensions $\log_{2-q}(0) = -1/(q-1)$ and $0\log_{2-q}(0) = 0$. As $q\to1$, the $q$-logarithm and $q$-exponential recover the usual logarithm and exponential. Using the $q$-logarithm, one can write $D_q$ as follows:
\begin{equation*}
D_q(\mu, \nu)
  = \int \frac{d\mu}{d\nu} \log_{2 - q}\left(\frac{d\mu}{d\nu}\right)\, d\nu
\end{equation*}
for $\mu \ll \nu$ and $q > 1$. 

Using properties of the $q$-exponential function and a direct computation, one can derive a variational identity relating $\mathrm{OT}_{q,\ep}$ and $D_q(\pi, K_{q,\ep})$.
\begin{prop}\label{prop1}
Let $q>1$ and $\ep>0$. Then
\begin{equation}\label{eq;prop1}
\mathrm{OT}_{q, \ep}(\mu_1, \mu_2) = \inf_{\substack{\pi \in \Pi(\mu_1, \mu_2)\\D_q(\pi,m)<+\infty}} \left\{\ep D_q(\pi, K_{q,\ep}) - (q - 1) \int_{X_1\times X_2} c \log_{2 - q}\left(\frac{d\pi}{dm}\right)\, d\pi\right\},
\end{equation}
where $K_{q,\ep}$ is the Tsallis kernel defined as
\begin{equation*}
  K_{q,\ep} = \exp_q\left(-\frac{c}{\ep}\right) m
  = \left[1+(q-1)\frac{c}{\ep}\right]^{-\frac{1}{q-1}} m.  
\end{equation*}
\end{prop}

Classical projection theory for $f$-divergences relates the well-posedness of information projection problems to the convexity of the generator and the geometry of the admissible set; see, e.g., \cite{Csiszar1975, BroniatowskiKeziou}. In our setting, the Tsallis divergence $D_q$ is generated by
\begin{equation*}
\phi_q(x) = x\log_{2-q}x = \frac{x^q-x}{q-1}.
\end{equation*}
For the dual formulation, it is convenient to use the affine normalization
\begin{equation*}
\Phi_q(x) = \phi_q(x)-x+1.
\end{equation*}
Since $m = \mu_1\otimes\mu_2$ and $\pi\in\Pi(\mu_1,\mu_2)$ are probability measures, the affine terms cancel after integration, and hence
\begin{equation*}
\int_{X_1\times X_2} \Phi_q\left(\frac{d\pi}{dm}\right)\,dm = D_q(\pi,m).
\end{equation*}

We now consider the information projection problem with respect to the Tsallis divergence $D_q$:
\begin{equation}\label{Tsallis projection pb}
\ep\inf_{\pi\in\Pi(\mu_1,\mu_2)} D_q(\pi,K_{q,\ep})
\end{equation}
for $1 < q \leq 2$. In what follows, we call this problem the Tsallis projection problem. We first note that this problem has a finite admissible competitor. Indeed, since $c$ is bounded and nonnegative,
\begin{equation*}
\left( 1+(q-1)\frac{\|c\|_\infty}{\ep} \right)^{-\frac{1}{q-1}} \leq \frac{dK_{q,\ep}}{dm} \leq 1.
\end{equation*}
Thus, $K_{q,\ep}$ and $m$ are mutually absolutely continuous, and both density ratios are bounded. Since $m\in\Pi(\mu_1,\mu_2)$, it follows that $D_q(m,K_{q,\ep}) < +\infty$. Moreover, whenever $D_q(\pi,K_{q,\ep})<+\infty$, we have
\begin{equation*}
D_q(\pi,K_{q,\ep}) \geq -\frac{1}{q-1}.
\end{equation*}
Consequently, the infimum in \eqref{Tsallis projection pb} is finite.

We next establish existence and uniqueness of the minimizer. The set $\Pi(\mu_1,\mu_2)$ is narrowly compact, while the functional $\pi\longmapsto D_q(\pi,K_{q,\ep})$ is lower semicontinuous with respect to narrow convergence. Therefore, the infimum in \eqref{Tsallis projection pb} is attained. Moreover, since $\phi_q$ is strictly convex for $q>1$, the Tsallis divergence is strictly convex among couplings for which it is finite. Indeed, let $\pi_0,\pi_1\in\Pi(\mu_1,\mu_2)$ be distinct and satisfy $D_q(\pi_0,K_{q,\ep})<+\infty$ and $D_q(\pi_1,K_{q,\ep})<+\infty$. Then, for every $t\in(0,1)$,
\begin{equation*}
D_q\left(t\pi_0+(1-t)\pi_1,K_{q,\ep}\right) < tD_q(\pi_0,K_{q,\ep}) + (1-t)D_q(\pi_1,K_{q,\ep}).
\end{equation*}
Since $\Pi(\mu_1,\mu_2)$ is convex, $t\pi_0+(1-t)\pi_1 \in \Pi(\mu_1,\mu_2)$.
Therefore, two distinct minimizers of \eqref{Tsallis projection pb} cannot exist, and hence the minimizer is unique.

Tsallis and more general divergence regularizations of optimal transport have been studied from several viewpoints \cite{Bao2022, muzellec2017, terjek2022optimal}. Our previous work \cite{suguro2023convergence} studies the limit $\ep\downarrow0$ toward unregularized optimal transport, whereas the present paper fixes $\ep > 0$ and considers the limit $q \downarrow 1$ toward KL-regularized optimal transport. 

The aim of this paper is to establish qualitative and quantitative convergence results for this Tsallis-to-KL limit. First, we prove that the Tsallis-regularized functionals \eqref{Tsallis-reg OT} $\Gamma$-converge to the KL-regularized functional \eqref{KL-reg OT}. We refer the reader to the monographs \cite{braides2002gamma,dal2012introduction} for background on $\Gamma$-convergence.

\begin{dfn}[$\Gamma$-convergence]\label{Gamma-convergence}
Let $\{q_k\}_{k\in\N}\subset(1,2]$ satisfy $q_k\downarrow1$. Define $F_k:\mathcal P(X_1\times X_2)\to\R\cup\{+\infty\}$ by
\begin{equation*}
F_k(\pi) = 
\begin{cases} 
\displaystyle \int_{X_1 \times X_2} c\, d\pi + \ep D_{q_k}(\pi, m) \quad &\text{if} \,\, \pi \in \Pi(\mu_1, \mu_2), \\
+\infty &\text{otherwise}, 
\end{cases}
\end{equation*}
and define $F:\mathcal P(X_1\times X_2)\to\R\cup\{+\infty\}$ by
\begin{equation*}
F(\pi) = 
\begin{cases} 
\displaystyle \int_{X_1 \times X_2} c\, d\pi + \ep \mathrm{KL}(\pi, m)\quad &\text{if} \,\, \pi \in \Pi(\mu_1, \mu_2), \\
+\infty &\text{otherwise}. 
\end{cases}
\end{equation*}
We say that $\{F_k\}_{k\in\N}$ $\Gamma$-converges to $F$ with respect to the narrow topology if the following two conditions hold:
\begin{description}
\item[(liminf inequality)] For any sequence $\{\pi_k\}_{k\in\N}\subset\mathcal P(X_1\times X_2)$ such that $\pi_k\to\pi$ narrowly,
\begin{equation*}
F(\pi) \leq \liminf_{k \to \infty} F_k(\pi_k).
\end{equation*}

\item[(limsup inequality)] For any $\pi\in\mathcal P(X_1\times X_2)$, there exists a sequence $\{\pi_k\}_{k\in\N}\subset\mathcal P(X_1\times X_2)$ such that $\pi_k\to\pi$ narrowly and
\begin{equation*}
F(\pi) \geq \limsup_{k \to \infty} F_k(\pi_k).
\end{equation*}
\end{description}
Here, $\mathcal P(X)$ denotes the set of Borel probability measures on $X$. 
\end{dfn}
The following theorem gives the qualitative convergence result.
\begin{thm}\label{thm;narrow}
Let $X_i$ be Polish spaces, let $\mu_i\in\mathcal P(X_i)$ $(i=1,2)$, and let $c\in C_b(X_1\times X_2)$ be nonnegative. Let $\ep>0$ and $\{q_k\}_{k\in\N}\subset(1,2]$ satisfy $q_k\downarrow1$. Then $\{F_k\}_{k\in\N}$ $\Gamma$-converges to $F$ with respect to the narrow topology. Moreover, the unique minimizer $\pi_k^*$ of $F_k$ converges narrowly to the unique minimizer $\pi^*$ of $F$.
\end{thm}

Next, we establish quantitative estimates for the Tsallis-to-KL limit as $q\downarrow1$. The following two theorems give quantitative error estimates of order $O(q-1)$ as $q\downarrow1$. The first concerns the regularized optimal transport value, whereas the second concerns the corresponding information projection value.

\begin{thm}\label{thm;2}
Let $X_i$ be Polish spaces, let $\mu_i\in\mathcal P(X_i)$ $(i=1,2)$, and let $c\in C_b(X_1\times X_2)$ be nonnegative. Let $\ep>0$ and $1<q\leq2$. Then
\begin{equation*}
0\leq\mathrm{OT}_{q, \ep}(\mu_1, \mu_2) - \mathrm{OT}_\ep(\mu_1, \mu_2) \leq \ep H_\ep(q - 1),
\end{equation*}
where $H_\ep\geq0$ is defined by
\begin{equation*}
H_\ep = \max \left\{ \frac{25 \|c\|_\infty^2}{2\ep^2}, \left( \exp\left(\frac{5\|c\|_\infty}{\ep}\right)-1 \right)\frac{5 \|c\|_\infty}{\ep}
\right\}. 
\end{equation*}
\end{thm}

\begin{thm}\label{thm;3}
Under the assumptions of Theorem~\ref{thm;2},
\begin{align*}
\left|\ep \inf_{\pi \in \Pi(\mu_1, \mu_2)} D_q(\pi, K_{q,\ep})
-\ep \inf_{\pi \in \Pi(\mu_1, \mu_2)} \mathrm{KL}(\pi, K_{\ep})\right|
\le C_\ep (q-1),
\end{align*}
where the $q$-independent constant $C_\ep\geq0$ is defined by
\begin{equation*}
C_\ep = \ep H_\ep+\|c\|_\infty\exp\left(\frac{5\|c\|_\infty}{\ep}\right).
\end{equation*}
\end{thm}

The paper is organized as follows: Section~\ref{pre} recalls the basic properties of the $q$-logarithm and $q$-exponential functions and the Schr\"odinger-potential estimates used below. In Section~\ref{sec:Gamma convergence}, we prove $\Gamma$-convergence of the Tsallis-regularized functionals to the KL-regularized functional and narrow convergence of the minimizers as $q\downarrow1$. In Section~\ref{sec:error estimate}, we establish error estimates of order $O(q-1)$ and prove Theorems~\ref{thm;2} and~\ref{thm;3}.

\section{Preliminaries}\label{pre}

This section collects basic properties of the $q$-logarithm and $q$-exponential functions and bounds for the associated Schr\"odinger potentials.

\subsection{Basic properties of the $q$-logarithm and $q$-exponential functions}
We recall basic properties of the $q$-logarithm and $q$-exponential functions, which play a central role in Tsallis statistical mechanics \cite{Naudts2011, tsallisbook2009} and in the information-geometric study of deformed exponential families \cite{AmariOhara2011}. The $q$-logarithm and $q$-exponential functions \eqref{q-log and q-exp} satisfy the following elementary identities whenever the expressions are defined. In particular, for $x\in I_q$,
\begin{equation*}
\log_q(\exp_q x)=x,\qquad \log_{2-q}\bigl((\exp_q x)^{-1}\bigr)=-x.
\end{equation*}
For $x,y>0$, the $q$-logarithm has the non-additivity property
\begin{equation*}
\log_q(xy)
= \log_q x+\log_q y+(1-q)\log_q x\,\log_q y.
\end{equation*}
Equivalently,
\begin{equation}\label{eq;non-additivity property}
\log_{2-q}(xy) = \log_{2-q}x+\log_{2-q}y+(q-1)\log_{2-q}x\,\log_{2-q}y.
\end{equation}

The Tsallis relative entropy can be expressed in terms of the $q$-logarithm as
\begin{equation*}
D_q(\mu,\nu)
=
\int \frac{d\mu}{d\nu}\,\log_{2-q}\!\left(\frac{d\mu}{d\nu}\right)\,d\nu,
\end{equation*}
whenever $\mu\ll \nu$ and $q>1$. This representation is particularly convenient in the present setting.

Combining the identity
\begin{equation*}
\log_{2-q}\bigl(\exp_q(-x)^{-1}\bigr)=x
\qquad (-x\in I_q)
\end{equation*}
with the non-additivity formula \eqref{eq;non-additivity property} yields the relation between the Tsallis projection problem \eqref{Tsallis projection pb} and the Tsallis-regularized optimal transport problem \eqref{Tsallis-reg OT}.

\begin{pf}{Proposition \ref{prop1}}
Fix $\pi\in\Pi(\mu_1,\mu_2)$ such that $D_q(\pi,m)<+\infty$. Then $\pi\ll m$, and we set
\begin{equation*}
\rho = \frac{d\pi}{dm}.
\end{equation*}
Since $D_q(\pi,m)<+\infty$, we have
\begin{equation*}
\int_{X_1\times X_2}\rho^q\,dm = 1+(q-1)D_q(\pi,m) <+\infty.
\end{equation*}
By the definition of the Tsallis kernel,
\begin{equation*}
\frac{dK_{q,\ep}}{dm} = \left(1+(q-1)\frac{c}{\ep}\right)^{-\frac{1}{q-1}}, \qquad \frac{dm}{dK_{q,\ep}} = \left(1+(q-1)\frac{c}{\ep}\right)^{\frac{1}{q-1}}.
\end{equation*}
Since $c$ is bounded and nonnegative, both density ratios are bounded. Hence $\pi\ll K_{q,\ep}$ and
\begin{align*}
\int_{X_1\times X_2} \left(\frac{d\pi}{dK_{q,\ep}}\right)^q\,dK_{q,\ep}
&= \int_{X_1\times X_2} \rho^q \left(\frac{dm}{dK_{q,\ep}}\right)^{q-1}\,dm
<+\infty.
\end{align*}
Therefore, $D_q(\pi,K_{q,\ep})<+\infty$. Moreover,
\begin{align*}
\int_{X_1\times X_2} \left|\log_{2-q}\rho\right|\,d\pi
&= \frac{1}{q-1} \int_{X_1\times X_2} \rho\left|\rho^{q-1}-1\right|\,dm \\
&\leq \frac{1}{q-1} \int_{X_1\times X_2} \left(\rho^q+\rho\right)\,dm <+\infty.
\end{align*}
Since $c$ is bounded, the correction term in \eqref{eq;prop1} is also finite. Thus, all the terms below are well defined.

We have
\begin{equation*}
\log_{2-q}\left(\frac{dm}{dK_{q,\ep}}\right) = \log_{2-q}\left( \left(1+(q-1)\frac{c}{\ep}\right)^{\frac{1}{q-1}} \right) = \frac{c}{\ep}.
\end{equation*}
Furthermore,
\begin{equation*}
\frac{d\pi}{dK_{q,\ep}} = \frac{d\pi}{dm}\frac{dm}{dK_{q,\ep}}.
\end{equation*}
Applying the non-additivity formula \eqref{eq;non-additivity property} on the set $\{\rho>0\}$, which has full $\pi$-measure, we obtain
\begin{align*}
\ep D_q(\pi,K_{q,\ep}) &= \ep\int_{X_1\times X_2} \log_{2-q}\left(\frac{d\pi}{dK_{q,\ep}}\right)\,d\pi \\
&= \ep\int_{X_1\times X_2} \log_{2-q}\left(\frac{d\pi}{dm}\right)\,d\pi + \ep\int_{X_1\times X_2} \log_{2-q}\left(\frac{dm}{dK_{q,\ep}}\right)\,d\pi \\
&\quad + \ep(q-1)\int_{X_1\times X_2} \log_{2-q}\left(\frac{d\pi}{dm}\right) \log_{2-q}\left(\frac{dm}{dK_{q,\ep}}\right)\,d\pi \\
&= \ep D_q(\pi,m) + \int_{X_1\times X_2}c\,d\pi + (q-1)\int_{X_1\times X_2}
c\log_{2-q}\left(\frac{d\pi}{dm}\right)\,d\pi.
\end{align*}
Rearranging this identity yields
\begin{equation*}
\int_{X_1\times X_2}c\,d\pi+\ep D_q(\pi,m)
= \ep D_q(\pi,K_{q,\ep}) - (q-1)\int_{X_1\times X_2} c\log_{2-q}\left(\frac{d\pi}{dm}\right)\,d\pi.
\end{equation*}

This identity holds for every $\pi\in\Pi(\mu_1,\mu_2)$ satisfying $D_q(\pi,m)< +\infty$. If $D_q(\pi,m) = +\infty$, then the objective functional in \eqref{Tsallis-reg OT} equals $+\infty$. Moreover, $m\in\Pi(\mu_1,\mu_2)$, $D_q(m,m)=0$, and $c$ is bounded, so \eqref{Tsallis-reg OT} has a finite admissible competitor. Therefore, restricting the infimum to couplings satisfying $D_q(\pi,m) < +\infty$ does not change its value. Taking the infimum in the preceding identity, we obtain
\begin{equation*}
\mathrm{OT}_{q,\ep}(\mu_1,\mu_2) = \inf_{\substack{\pi\in\Pi(\mu_1,\mu_2)\\D_q(\pi,m)<+\infty}} \left\{ \ep D_q(\pi,K_{q,\ep}) - (q-1)\int_{X_1\times X_2} c\log_{2-q}\left(\frac{d\pi}{dm}\right)\,d\pi \right\}.
\end{equation*}
This proves \eqref{eq;prop1}.
\end{pf}

\subsection{Schr\"odinger potentials for Tsallis-regularized optimal transport}

In this subsection, we recall the dual formulation of the Tsallis-regularized optimal transport problem \eqref{Tsallis-reg OT} and basic properties of the associated Schr\"odinger potentials. These facts will be used later to compare Tsallis and KL regularization. 

For a convex function $f$, we define its Legendre transform $f^*$ by
\begin{equation*}
f^*(y) = \sup_{x \geq 0} \{x y - f(x)\}
\qquad (y \in \R).
\end{equation*}
If we set
\begin{equation*}
\Phi_q(x) = \phi_q(x)-x+1=\frac{x^q - x}{q - 1} - x + 1\quad (x \geq 0),
\end{equation*}
then $\Phi_q(1)=\Phi_q'(1)=0$ and
\begin{align*}
  &\int_{X_1 \times X_2} \Phi_q\left( \frac{d\pi}{dm} \right)\, dm \\
  &\quad = \int_{X_1 \times X_2} \left(\frac{\left(\frac{d\pi}{dm}\right)^q - \frac{d\pi}{dm}}{q - 1} - \frac{d\pi}{dm} + 1\right)\, dm \\
  &\quad 
  = D_q(\pi, m).
\end{align*}
Moreover, the Legendre transform $\Psi_q=\Phi_q^*$ and its derivative are given by
\begin{equation*}
\Psi_q(y) = \left(\frac{q-1}{q}y+1\right)_+^{\frac{q}{q-1}}-1, \qquad 
\Psi_q'(y) = \left(\frac{q-1}{q}y+1\right)_+^{\frac{1}{q-1}},
\end{equation*}
where $z_+ =\max\{z,0\}$. 

With this notation, the Tsallis-regularized optimal transport problem is written as
\begin{equation*}
\mathrm{OT}_{q, \ep}(\mu_1, \mu_2) = \inf_{\pi \in \Pi(\mu_1, \mu_2)}
\left\{ \int_{X_1 \times X_2} c\, d\pi + \ep \int_{X_1 \times X_2} \Phi_q\left( \frac{d\pi}{dm} \right)\, dm \right\},
\end{equation*}
and its dual problem is given by
\begin{equation}\label{prob;dual}
\sup_{\substack{u \in L^{q'}(\mu_1), \\ v \in L^{q'}(\mu_2)}}
\left\{
\int_{X_1} u\, d\mu_1
+ \int_{X_2} v\, d\mu_2
- \ep \int_{X_1 \times X_2}
\Psi_q\left( \frac{u(x) + v(y) - c(x,y)}{\ep} \right)\, dm
\right\},
\end{equation}
where we set $q' = q/(q - 1)$. 
We call a maximizer $(u_q, v_q)$ of \eqref{prob;dual} a pair of Schr\"odinger potentials associated with \eqref{Tsallis-reg OT}.

We next recall the existence of dual maximizers and the representation of the optimal coupling. 

\begin{thm}[{\cite[Lemma 3.3, Theorem 3.4]{di2020optimal}}]\label{thm;estimates of potentials}
Let $1<q\leq2$, let $c:X_1\times X_2\to\R$ be bounded, and let $\ep>0$. Then the dual problem \eqref{prob;dual} admits a maximizer. Moreover, one can choose a maximizer satisfying
\begin{equation*}
\|u_q\|_{\infty}, \ \|v_q\|_{\infty} \leq 2 \|c\|_{\infty}.
\end{equation*}
\end{thm}

\begin{thm}[{\cite[Proposition 3.5, Theorem 3.6]{di2020optimal}}]
Let $1<q\leq2$ and let $(u_q,v_q)$ be a maximizer of the dual problem \eqref{prob;dual}. Then the optimal coupling $\pi_{q,\ep}^*$ for \eqref{Tsallis-reg OT} is represented by
\begin{equation}\label{eq:piq-density}
\frac{d\pi_{q, \ep}^*}{dm} = \Psi_q'\left( \frac{u_q(x) + v_q(y) - c(x,y)}{\ep} \right) = \left(\frac{q - 1}{q} \frac{u_q(x) + v_q(y) - c(x,y)}{\ep} + 1\right)_+^{\frac{1}{q - 1}}.
\end{equation}
\end{thm}
The above representation immediately yields an $L^\infty$ bound on $d\pi_{q, \ep}^*/dm$ that is uniform in $q$.
\begin{prop}\label{prop:uniform-Linfty-density}
For all $1 < q \leq 2$, the following uniform $L^\infty$ bound holds:
\begin{equation*}
\left\| \frac{d\pi_{q, \ep}^*}{dm} \right\|_{\infty}
\leq \left(1+\frac{5 (q - 1)}{q \ep} \|c\|_{\infty}\right)^{\frac{1}{q - 1}}
\leq \exp\left(\frac{5\|c\|_\infty}{\ep}\right).
\end{equation*}
\end{prop}

\begin{pf}{Proposition \ref{prop:uniform-Linfty-density}}
By Theorem \ref{thm;estimates of potentials}, we have
\begin{equation*}
|u_q(x)| \leq 2 \|c\|_{\infty}, \quad |v_q(y)| \leq 2 \|c\|_{\infty}, \quad \mathrm{and}\quad |c(x,y)| \leq \|c\|_{\infty}.
\end{equation*}
Hence,
\begin{equation*}
u_q(x) + v_q(y) - c(x,y) \leq 5 \|c\|_{\infty}.
\end{equation*}
Therefore, by \eqref{eq:piq-density},
\begin{equation*}
\frac{d\pi_{q, \ep}^*}{dm} \leq \left(1+\frac{q - 1}{q} \frac{5 \|c\|_{\infty}}{\ep}\right)^{\frac{1}{q - 1}} \leq\exp\left(\frac{5\|c\|_\infty}{\ep}\right),
\end{equation*}
where we used $1+t\leq e^t$ for $t\geq0$ and $q>1$.
\end{pf}

\begin{rem}
For comparison, in the KL case, we have
\begin{equation*}
\Phi(x) = x \log x -x + 1,
\qquad
\Psi(y) = e^y - 1.
\end{equation*}
The optimal coupling $\pi_\ep^*$ for \eqref{KL-reg OT} is represented by
\begin{equation*}
\frac{d\pi_\ep^*}{dm}
= \exp\left( \frac{u_\ep(x) + v_\ep(y) - c(x,y)}{\ep} \right),
\end{equation*}
where $(u_\ep, v_\ep)$ is a pair of Schr\"odinger potentials for $\mathrm{OT}_\ep$. The KL case admits the same normalization $\|u_\ep\|_\infty,\|v_\ep\|_\infty\leq2\|c\|_\infty$; see, e.g., \cite{di2020optimal}. Arguing as in the proof of Proposition~\ref{prop:uniform-Linfty-density}, we obtain
\begin{equation*}
\exp\left(-\frac{5\|c\|_\infty}{\ep}\right)
\leq\frac{d\pi_\ep^*}{dm}
\leq\exp\left(\frac{5\|c\|_\infty}{\ep}\right)
\quad m\text{-a.e.}
\end{equation*}
Therefore, both the Tsallis- and KL-regularized optimal couplings satisfy the same upper $L^\infty$ bound, while the KL density also has a positive lower bound.
\end{rem}

\section{$\Gamma$-convergence from Tsallis to KL regularization}\label{sec:Gamma convergence}

\begin{pf}{Theorem \ref{thm;narrow}}
Let $\{q_k\}_{k \in \N} \subset (1,2]$ satisfy $q_k\downarrow1$. We define the functionals by
\begin{equation*}
F_k(\pi) = \left\{
\begin{aligned} &\int_{X_1 \times X_2} c\, d\pi + \ep D_{q_k}(\pi, m)
&& \mathrm{for}\ \pi \in \Pi(\mu_1, \mu_2), \\
&+ \infty, && \mathrm{otherwise},
\end{aligned}
\right.
\end{equation*}
and
\begin{equation*}
F(\pi) = \left\{
\begin{aligned}
&\int_{X_1 \times X_2} c\, d\pi + \ep \mathrm{KL}(\pi, m)
&& \mathrm{for}\ \pi \in \Pi(\mu_1, \mu_2), \\
&+ \infty, && \mathrm{otherwise}.
\end{aligned}
\right.
\end{equation*}

First, we prove the liminf inequality
\begin{equation}\label{liminf;0}
  F(\pi)
  \leq \liminf_{k \to \infty} F_k(\pi_k)
\end{equation}
for any $\{\pi_k\}_{k \in \N}$ such that $\pi_k \to \pi$ narrowly. If $\displaystyle \liminf_{k\to\infty}F_k(\pi_k)=+\infty$, there is nothing to prove. Otherwise, after passing to a subsequence realizing the liminf, we may assume that $F_k(\pi_k)<+\infty$ for every $k$. Hence $\pi_k\in\Pi(\mu_1,\mu_2)$, and the narrow closedness of $\Pi(\mu_1,\mu_2)$ implies that $\pi\in\Pi(\mu_1,\mu_2)$. Moreover, for every probability measure $\eta$,
\begin{equation*}
\mathrm{KL}(\eta,m) \leq D_{q_k}(\eta,m).
\end{equation*}
Indeed, when $\eta \ll m$, this follows from $\log x\leq\log_{2-q_k}x$ for $x>0$, while otherwise both sides are $+\infty$. Consequently, the lower semicontinuity of the Kullback--Leibler divergence gives
\begin{equation}\label{liminf;1}
  \ep \mathrm{KL}(\pi, m)
  \leq \liminf_{k \to \infty} \ep \mathrm{KL}(\pi_k, m)
  \leq \liminf_{k \to \infty} \ep D_{q_k}(\pi_k, m).
\end{equation}
On the other hand, by the lower semicontinuity of the transport cost functional,
\begin{equation}\label{liminf;2}
\int_{X_1 \times X_2} c\, d\pi \leq \liminf_{k \to \infty} \int_{X_1 \times X_2} c\, d\pi_k.
\end{equation}
Combining \eqref{liminf;1} and \eqref{liminf;2}, and using the superadditivity of the $\liminf$, we obtain
\begin{equation*}
F(\pi) \leq \liminf_{k\to\infty} F_k(\pi_k),
\end{equation*}
which proves \eqref{liminf;0}.

Next, we prove the limsup inequality. For any $\pi\in\mathcal P(X_1\times X_2)$, we seek a sequence $\{\pi_k\}_{k\in\N}$ converging narrowly to $\pi$ such that
\begin{equation}\label{limsup;0}
  F(\pi) \geq \limsup_{k \to \infty} F_k(\pi_k).
\end{equation}
It is enough to consider the case $F(\pi)<+\infty$, since otherwise the constant sequence $\pi_k = \pi$ satisfies the limsup inequality. Thus $\pi\in\Pi(\mu_1,\mu_2)$ and $d\pi=\rho\,dm$ for some density $\rho$.
For $n \in \N$, we set $\tilde{\rho}_n = \min\{\rho, n\}$, 
\begin{equation*}
  \dl_n = 1 - \int_{X_1 \times X_2} \tilde{\rho}_n\, dm,
\end{equation*}
\begin{equation*}
a_n(x) = 1 - \int_{X_2} \tilde{\rho}_n(x, y)\, d\mu_2(y), \quad \mathrm{and} \quad b_n(y) = 1 - \int_{X_1} \tilde{\rho}_n(x, y)\, d\mu_1(x). 
\end{equation*}
We note that
\begin{equation*}
\int_{X_1} a_n(x)\, d\mu_1 = \int_{X_2} b_n(y)\, d\mu_2 = \dl_n.
\end{equation*}
We set 
\begin{equation*}
\rho_n(x, y) = \left\{
    \begin{aligned}
    &\tilde{\rho}_n(x, y) + \frac{a_n(x) b_n(y)}{\dl_n}
    && \mathrm{if}\ \dl_n > 0, \\
    &\rho
    && \mathrm{if}\ \dl_n = 0
    \end{aligned}
    \right.
\end{equation*}
and $d\pi_n =\rho_n\,dm$. Then $\pi_n \in \Pi(\mu_1, \mu_2)$. If $\dl_n>0$, since $\dl_n\to0$ by monotone convergence,
\begin{equation*}
\|\pi_n-\pi\|_{\mathrm{TV}} = \int_{X_1\times X_2}|\rho_n-\rho|\,dm \leq \int_{X_1\times X_2} |\rho-\tilde\rho_n| \,dm +\int_{X_1\times X_2}\frac{a_nb_n}{\dl_n}\,dm = 2\dl_n\to0. 
\end{equation*}
If $\dl_n=0$, then $\rho_n=\rho$ by definition. Therefore, $\pi_n\to\pi$ in total variation, and hence narrowly.
Moreover, since $\int_{X_1\times X_2}|\rho_n-\rho|\,dm\to0$, we have
\begin{equation}\label{limsup;2}
\left|\int_{X_1\times X_2}c\,d\pi_n-\int_{X_1\times X_2}c\,d\pi\right|
\leq \|c\|_\infty\int_{X_1\times X_2}|\rho_n-\rho|\,dm\to0.
\end{equation}
Furthermore, by the marginal conditions for $\rho$ and the inequality $\tilde{\rho}_n\leq\rho$, we have $0\leq a_n(x)\leq1$ for $\mu_1$-a.e. $x\in X_1$ and $0\leq b_n(y)\leq1$ for $\mu_2$-a.e. $y\in X_2$. Therefore, if $\dl_n>0$, then
\begin{equation}\label{limsup;1}
\rho_n(x,y)\leq\tilde{\rho}_n(x,y)+\frac{1}{\dl_n}\leq n+\frac{1}{\dl_n}\quad m\text{-a.e.}
\end{equation}
Since $\dl_n>0$, the right-hand side is finite for each fixed $n$, and hence $\rho_n\in L^\infty(m)$. If $\dl_n=0$, then
\begin{equation*}
0=\dl_n=\int_{X_1\times X_2}(\rho-\tilde{\rho}_n)\,dm.
\end{equation*}
Since $\rho-\tilde{\rho}_n\geq0$, it follows that $\rho=\tilde{\rho}_n$ $m$-a.e. By the definition of $\rho_n$, we then have
\begin{equation*}
\rho_n=\rho=\tilde{\rho}_n\leq n\quad m\text{-a.e.}
\end{equation*}
Thus, in either case, $\rho_n\in L^\infty(m)$ for every $n$.

In what follows, we show 
\begin{equation}\label{limsup;3}
\lim_{n\to\infty}\mathrm{KL}(\pi_n,m)=\mathrm{KL}(\pi,m).
\end{equation}
If $\dl_n=0$ for some $n$, then $\rho=\tilde{\rho}_n\leq n$ $m$-a.e. Hence, for every $j\geq n$, we have $\tilde{\rho}_j=\rho$, and therefore $\dl_j=0$ and $\pi_j=\pi$. In this case, \eqref{limsup;3} follows immediately. Otherwise, $\dl_n>0$ for every $n$. We then set
\begin{equation*}
r_n = \frac{\tilde{\rho}_n}{1-\dl_n} \quad\mathrm{and}\quad s_n = \frac{a_nb_n}{\dl_n^2}.
\end{equation*}
Since
\begin{equation*}
\int_{X_1\times X_2}r_n\,dm = 1 \quad\mathrm{and}\quad \int_{X_1\times X_2}s_n\,dm = 1,
\end{equation*}
both $r_n$ and $s_n$ are probability densities with respect to $m$. Moreover,
\begin{equation*}
\rho_n = (1-\dl_n)r_n + \dl_n s_n.
\end{equation*}
Applying Jensen's inequality to the convex function $t\mapsto t\log t$, we obtain
\begin{align*}
\mathrm{KL}(\pi_n,m)
&=\int_{X_1\times X_2}\rho_n\log\rho_n\,dm \\
&\leq (1-\dl_n)\int_{X_1\times X_2}r_n\log r_n\,dm
+\dl_n\int_{X_1\times X_2}s_n\log s_n\,dm \\
&=\int_{X_1\times X_2}\tilde{\rho}_n\log\tilde{\rho}_n\,dm
-(1-\dl_n)\log(1-\dl_n) \\
&\quad +\int_{X_1}a_n\log a_n\,d\mu_1
+\int_{X_2}b_n\log b_n\,d\mu_2
-2\dl_n\log\dl_n.
\end{align*}
Since $\tilde{\rho}_n\leq\rho$, we have
\begin{equation*}
(\tilde{\rho}_n\log\tilde{\rho}_n)_+ \leq (\rho\log\rho)_+.
\end{equation*}
Moreover, $t\log t\geq-e^{-1}$ for every $t\geq0$, where $0\log0=0$. Hence
\begin{equation*}
\left|\tilde{\rho}_n\log\tilde{\rho}_n\right| \leq (\rho\log\rho)_+ + e^{-1} \quad m\text{-a.e.}
\end{equation*}
Since $\mathrm{KL}(\pi,m)<+\infty$, the right-hand side is integrable with respect to $m$. Therefore, by the dominated convergence theorem, we have
\begin{equation*}
\lim_{n\to\infty}\int_{X_1\times X_2}\tilde{\rho}_n\log\tilde{\rho}_n\,dm = \int_{X_1\times X_2}\rho\log\rho\,dm.
\end{equation*}
Furthermore, by the marginal identities for $\rho$ and the monotone convergence theorem,
\begin{equation*}
\int_{X_2}\tilde{\rho}_n(x,y)\,d\mu_2(y)\uparrow1
\quad\text{for $\mu_1$-a.e. }x,
\end{equation*}
and
\begin{equation*}
\int_{X_1}\tilde{\rho}_n(x,y)\,d\mu_1(x)\uparrow1
\quad\text{for $\mu_2$-a.e. }y.
\end{equation*}
Hence $a_n\downarrow0$ $\mu_1$-a.e. and $b_n\downarrow0$ $\mu_2$-a.e. Since $0\leq a_n,b_n\leq1$ and $|t\log t|\leq e^{-1}$ on $[0,1]$, where $0\log0=0$, the dominated convergence theorem gives
\begin{equation*}
\lim_{n\to\infty}\int_{X_1}a_n\log a_n\,d\mu_1=0 \quad\mathrm{and}\quad
\lim_{n\to\infty}\int_{X_2}b_n\log b_n\,d\mu_2=0.
\end{equation*}
Since $\dl_n\to0$, we also have
\begin{equation*}
(1-\dl_n)\log(1-\dl_n)\to0 \quad\mathrm{and}\quad \dl_n\log\dl_n\to0.
\end{equation*}
Therefore,
\begin{equation*}
\limsup_{n\to\infty}\mathrm{KL}(\pi_n,m)\leq\mathrm{KL}(\pi,m).
\end{equation*}
On the other hand, since $\pi_n\to\pi$ narrowly and the Kullback--Leibler divergence is lower semicontinuous with respect to narrow convergence,
\begin{equation*}
\liminf_{n\to\infty}\mathrm{KL}(\pi_n,m)\geq\mathrm{KL}(\pi,m).
\end{equation*}
Consequently, we obtain \eqref{limsup;3}.

For every fixed $n$, the density $\rho_n$ is bounded. Let $M_n<+\infty$ be such that $0\leq\rho_n\leq M_n$ $m$-a.e. For $t>0$ and $1<q\leq2$, the mean value theorem applied to the function $s\mapsto t^s$ yields
\begin{equation*}
\frac{t^q-t}{q-1}=t^\xi\log t
\end{equation*}
for some $\xi\in(1,q)$. At points where $\rho_n=0$, both $(\rho_n^q-\rho_n)/(q-1)$ and $\rho_n\log\rho_n$ are equal to zero. Hence
\begin{equation*}
\frac{\rho_n^q-\rho_n}{q-1}\to\rho_n\log\rho_n \quad m\text{-a.e. as }q \downarrow 1.
\end{equation*}
Moreover, if $0<t\leq1$, then $|t^\xi\log t|\leq t|\log t|\leq e^{-1}$, while if $1\leq t\leq M_n$, then $|t^\xi\log t|\leq M_n^2\log(\max\{M_n,1\})$. Thus the integrands are dominated by an $m$-integrable constant independent of $q$. Therefore, by the dominated convergence theorem,
\begin{equation*}
D_q(\pi_n,m)\to\mathrm{KL}(\pi_n,m)
\quad\text{as }q\downarrow1.
\end{equation*}

We now choose a diagonal sequence so that $n\to\infty$ and $q_k\downarrow1$ simultaneously. For each $n\in\N$, there exists $K_n\in\N$ such that
\begin{equation*}
|D_{q_k}(\pi_n,m)-\mathrm{KL}(\pi_n,m)|<\frac{1}{n}
\end{equation*}
for every $k\geq K_n$. Increasing $K_n$ if necessary, we may assume that $K_1<K_2<\cdots$. Define
\begin{equation*}
n(k)=
\begin{cases}
1, & k<K_1,\\
\max\{n\in\N:K_n\leq k\}, & k\geq K_1.
\end{cases}
\end{equation*}
Then $n(k)\to\infty$ as $k\to\infty$. Set $\hat{\pi}_k=\pi_{n(k)}$. Since $n(k)\to\infty$ and $\pi_n\to\pi$ narrowly, we have $\hat{\pi}_k\to\pi$ narrowly. Moreover, for every $k\geq K_1$, we have $k\geq K_{n(k)}$, and hence
\begin{equation*}
|D_{q_k}(\hat{\pi}_k,m)-\mathrm{KL}(\hat{\pi}_k,m)|<\frac{1}{n(k)}.
\end{equation*}
Therefore,
\begin{equation*}
D_{q_k}(\hat{\pi}_k,m)\leq\mathrm{KL}(\hat{\pi}_k,m)+\frac{1}{n(k)}.
\end{equation*}
Since $n(k)\to\infty$, \eqref{limsup;3} gives
\begin{equation}\label{limsup;4}
\limsup_{k\to\infty}D_{q_k}(\hat{\pi}_k,m)\leq\mathrm{KL}(\pi,m).
\end{equation}
Finally, \eqref{limsup;2} also implies
\begin{equation*}
\int_{X_1\times X_2}c\,d\hat{\pi}_k\to\int_{X_1\times X_2}c\,d\pi.
\end{equation*}
Combining this convergence with \eqref{limsup;4} yields \eqref{limsup;0}.

Recall that the sequence $\{F_k\}$ is equicoercive on $\mathcal{P}(X_1\times X_2)$ if, for every $M\in\R$, there exists a narrowly compact set $K_M\subset\mathcal{P}(X_1\times X_2)$ such that
\begin{equation*}
\{\pi\in\mathcal{P}(X_1\times X_2):F_k(\pi)\leq M\}\subset K_M
\end{equation*}
for every $k\in\N$. Since $F_k(\pi)<+\infty$ implies $\pi\in\Pi(\mu_1,\mu_2)$, every sublevel set of $F_k$ is contained in $\Pi(\mu_1,\mu_2)$. Since $\Pi(\mu_1,\mu_2)$ is narrowly compact, we may take $K_M=\Pi(\mu_1,\mu_2)$ for every $M\in\R$. Hence the sequence $\{F_k\}$ is equicoercive. Each $F_k$, as well as $F$, is lower semicontinuous with respect to narrow convergence. Moreover, $m\in\Pi(\mu_1,\mu_2)$ has finite energy for all $F_k$ and $F$, and hence minimizers exist. Strict convexity of the entropy terms on the convex set $\Pi(\mu_1,\mu_2)$ gives uniqueness. By the fundamental theorem of $\Gamma$-convergence and the equicoercivity of $\{F_k\}$, every cluster point of $\{\pi_k^*\}_{k\in\N}$ is a minimizer of $F$. Since $F$ has the unique minimizer $\pi^*$, every cluster point coincides with $\pi^*$. Hence the whole sequence $\{\pi_k^*\}_{k\in\N}$ converges narrowly to $\pi^*$.
\end{pf}

\section{Error estimates for the Tsallis-to-KL limit}\label{sec:error estimate}
In this section, we derive estimates of order $O(q-1)$ for the difference between the Tsallis- and KL-regularized optimal transport values. We then estimate the difference between the corresponding information projection values.

\subsection{Error estimate for the regularized optimal transport values}
\begin{pf}{Theorem \ref{thm;2}}
We first consider the difference between the following two regularized optimal transport values:
\begin{align}\label{OT;1}
\mathrm{OT}_{q,\ep}(\mu_1, \mu_2) - \mathrm{OT}_{\ep}(\mu_1, \mu_2) \notag 
&= \inf_{\pi \in \Pi(\mu_1, \mu_2)} 
\left\{
\int_{X_1 \times X_2} c\, d\pi + \ep D_q(\pi, m)
\right\} \\
&\quad - \inf_{\pi \in \Pi(\mu_1, \mu_2)}
\left\{
\int_{X_1 \times X_2} c\, d\pi + \ep \mathrm{KL}(\pi, m)
\right\}.
\end{align}
Since
\begin{equation*}
D_q(\pi, m) \geq \mathrm{KL}(\pi, m)
\end{equation*}
for any $\pi \in \Pi(\mu_1,\mu_2)$, it follows immediately that
\begin{equation*}
\mathrm{OT}_{q,\ep}(\mu_1, \mu_2) - \mathrm{OT}_{\ep}(\mu_1, \mu_2) \geq 0.
\end{equation*}

Let $\pi_\ep^*$ be the optimal coupling for the KL-regularized optimal transport problem \eqref{KL-reg OT}. Then
\begin{equation}\label{OT;2}
  \mathrm{OT}_{\ep}(\mu_1, \mu_2)
  = \int_{X_1 \times X_2} c\, d\pi_\ep^* + \ep \mathrm{KL}(\pi_\ep^*, m).
\end{equation}
Substituting \eqref{OT;2} into \eqref{OT;1}, we obtain
\begin{align}
  0 &\leq \mathrm{OT}_{q,\ep}(\mu_1, \mu_2) - \mathrm{OT}_{\ep}(\mu_1, \mu_2) \notag \\
  &= \inf_{\pi \in \Pi(\mu_1, \mu_2)}
  \left\{
    \int_{X_1 \times X_2} c\, d\pi + \ep D_q(\pi, m)
  \right\} 
  - \int_{X_1 \times X_2} c\, d\pi_\ep^* - \ep \mathrm{KL}(\pi_\ep^*, m) \notag \\
  &\leq \int_{X_1 \times X_2} c\, d\pi_\ep^* + \ep D_q(\pi_\ep^*, m)
  - \int_{X_1 \times X_2} c\, d\pi_\ep^* - \ep \mathrm{KL}(\pi_\ep^*, m) \notag \\
  &= \ep (D_q(\pi_\ep^*, m) - \mathrm{KL}(\pi_\ep^*, m)). \label{OT;3}
\end{align}
We estimate the pointwise difference between the two entropy integrands.
Set $a = q-1 \in (0,1]$. Then
\begin{equation*}
\log_{2-q}(x) - \log x
= \frac{x^a - 1}{a} - \log x
= \frac{e^{a \log x} - 1 - a \log x}{a}.
\end{equation*}
If $0 < x < 1$, setting $u=-a\log x>0$, we have
\begin{equation*}
e^{a\log x}-1-a\log x
= e^{-u}-1+u
= \int_0^u (1-e^{-s})\,ds
\leq \int_0^u s\,ds = \dfrac{u^2}{2} = \dfrac{1}{2}a^2(\log x)^2.
\end{equation*}
If $x \geq 1$, setting $t=\log x \geq 0$, we have
\begin{equation*}
e^{a\log x}-1-a\log x
= \int_0^{at}(e^s-1)\,ds
\leq at(e^{at}-1).
\end{equation*}
Moreover,
\begin{equation*}
e^{at}-1
= a\int_0^t e^{ar}\,dr
\leq a\int_0^t e^r\,dr
= a(e^t-1)
= a(x-1).
\end{equation*}
Hence,
\begin{equation*}
e^{a\log x}-1-a\log x
\leq a^2(x-1)\log x.
\end{equation*}
Dividing both sides by $a=q-1$, we obtain
\begin{equation}\label{ineq;Tsallis-to-KL ineq}
\log_{2-q}(x)-\log x
\leq (q-1)h(x),
\end{equation}
where
\begin{equation*}
h(x)
=
\left\{
\begin{aligned}
&\dfrac{1}{2}(\log x)^2, && 0<x<1,\\
&(x-1)\log x, && x\geq1.
\end{aligned}
\right.
\end{equation*}
Using the inequality \eqref{ineq;Tsallis-to-KL ineq}, we have 
\begin{equation*}
D_q(\pi_\ep^*, m) - \mathrm{KL}(\pi_\ep^*, m)
  \leq (q - 1) \int_{X_1 \times X_2} h\left(\frac{d \pi_\ep^*}{dm}\right)\, d\pi_\ep^*.
\end{equation*}
Since $\pi_\ep^*$ is the optimal coupling for \eqref{KL-reg OT}, it admits the representation
\begin{equation*}
 \frac{d\pi_\ep^*}{dm} = \exp \left( \frac{u_\ep(x) + v_\ep(y) - c(x, y)}{\ep} \right),
\end{equation*}
where $(u_\ep,v_\ep)$ is a normalized pair of Schr\"odinger potentials satisfying $\|u_\ep\|_\infty,\|v_\ep\|_\infty\leq2\|c\|_\infty$. Since $0\leq c\leq\|c\|_\infty$, it follows that
\begin{equation*}
\exp\left(-\frac{5\|c\|_\infty}{\ep}\right) \leq \frac{d\pi_\ep^*}{dm} \leq \exp\left(\frac{5\|c\|_\infty}{\ep}\right).
\end{equation*}
Thus, if we define
\begin{equation*}
H_\ep = \max \left\{ \frac{25 \|c\|_\infty^2}{2\ep^2}, \left( \exp\left(\frac{5\|c\|_\infty}{\ep}\right)-1 \right)\frac{5 \|c\|_\infty}{\ep}
\right\},
\end{equation*}
then we have
\begin{equation*}
\int_{X_1 \times X_2} h\left(\frac{d \pi^*_\ep}{dm}\right)\, d\pi^*_\ep \leq H_\ep, 
\end{equation*}
and hence we obtain
\begin{equation}\label{OT;6}
  D_q(\pi_\ep^*, m) - \mathrm{KL}(\pi_\ep^*, m) \leq (q - 1) H_\ep.
\end{equation}
Combining \eqref{OT;3} and \eqref{OT;6}, we obtain
\begin{equation*}
0 \leq \mathrm{OT}_{q,\ep}(\mu_1, \mu_2) - \mathrm{OT}_{\ep}(\mu_1, \mu_2)
\leq \ep (q - 1) H_\ep.
\end{equation*}
\end{pf}

\subsection{Error estimate for the information projection values}
\begin{pf}{Theorem \ref{thm;3}}

We recall that the information-geometric characterization of \eqref{KL-reg OT} is given by
\begin{equation}\label{KL;1}
  \mathrm{OT}_\ep(\mu_1, \mu_2)
  = \ep \inf_{\pi \in \Pi(\mu_1, \mu_2)} \mathrm{KL}(\pi, K_\ep),
\end{equation}
where $K_\ep = e^{-c/\ep} m$. On the other hand, by \eqref{eq;prop1}, the Tsallis-regularized optimal transport problem can be written as
\begin{equation}\label{KL;1q}
  \mathrm{OT}_{q,\ep}(\mu_1, \mu_2)
  = \ep D_q(\pi_{q,\ep}^*, K_{q,\ep})
    - (q-1)\int_{X_1 \times X_2}
    c \log_{2-q}\left(\frac{d\pi_{q,\ep}^*}{dm}\right)\, d\pi_{q,\ep}^*,
\end{equation}
where $\pi_{q,\ep}^*$ is the optimal coupling for \eqref{Tsallis-reg OT}. In this subsection, we estimate the error term
\begin{equation*}
  \inf_{\pi \in \Pi(\mu_1, \mu_2)} D_q(\pi, K_{q,\ep})
  -
  \inf_{\pi \in \Pi(\mu_1, \mu_2)} \mathrm{KL}(\pi, K_\ep).
\end{equation*}
Since
\begin{equation*}
  \ep \inf_{\pi \in \Pi(\mu_1, \mu_2)} D_q(\pi, K_{q,\ep})
  \leq \ep D_q(\pi_{q,\ep}^*, K_{q,\ep}),
\end{equation*}
it follows from \eqref{KL;1} and \eqref{KL;1q} that
\begin{equation}\label{KL;2}
\begin{split}
  &\ep \inf_{\pi \in \Pi(\mu_1, \mu_2)} D_q(\pi, K_{q,\ep})
  - \ep \inf_{\pi \in \Pi(\mu_1, \mu_2)} \mathrm{KL}(\pi, K_\ep) \\
  &\quad \leq \mathrm{OT}_{q,\ep}(\mu_1, \mu_2) - \mathrm{OT}_\ep(\mu_1, \mu_2)
  + (q-1)\int_{X_1 \times X_2}
    c \log_{2-q}\left(\frac{d\pi_{q,\ep}^*}{dm}\right)\, d\pi_{q,\ep}^*.
\end{split}
\end{equation}

We next estimate the second term on the right-hand side of \eqref{KL;2}. Since $c$ is a nonnegative bounded cost function, we have
\begin{align*}
\int_{X_1 \times X_2} c \log_{2-q}\left(\frac{d\pi_{q,\ep}^*}{dm}\right)\, d\pi_{q,\ep}^* 
&\leq \int_{\left\{ \frac{d\pi_{q,\ep}^*}{dm} \geq 1 \right\}} c \log_{2-q}\left(\frac{d\pi_{q,\ep}^*}{dm}\right)\, d\pi_{q,\ep}^* \\ 
&\leq \|c\|_\infty \int_{\left\{ \frac{d\pi_{q,\ep}^*}{dm} \geq 1 \right\}} \log_{2-q}\left(\frac{d\pi_{q,\ep}^*}{dm}\right)\,d\pi_{q,\ep}^*.
\end{align*}
We set
\begin{equation*}
D_q^+(\pi_{q,\ep}^*, m) = \int_{\left\{
\frac{d\pi_{q,\ep}^*}{dm} \geq 1
\right\}}
\log_{2-q}\left(\frac{d\pi_{q,\ep}^*}{dm}\right)\, d\pi_{q,\ep}^*.
\end{equation*}
Then
\begin{equation}\label{KL;3}
  \int_{X_1 \times X_2}
  c \log_{2-q}\left(\frac{d\pi_{q,\ep}^*}{dm}\right)\, d\pi_{q,\ep}^*
  \leq \|c\|_\infty D_q^+(\pi_{q,\ep}^*, m).
\end{equation}
Next, we estimate $D_q^+(\pi_{q,\ep}^*, m)$. Since $1<q\leq 2$, we have
\begin{equation*}
  \log_{2-q}(x)
  = \frac{x^{q-1}-1}{q-1}
  \leq x-1
  \qquad (x \geq 1).
\end{equation*}
Hence,
\begin{align*}
  D_q^+(\pi_{q,\ep}^*, m)
  &\leq \int_{\left\{
  \frac{d\pi_{q,\ep}^*}{dm} \geq 1
  \right\}}
  \left(
  \frac{d\pi_{q,\ep}^*}{dm} - 1
  \right)\, d\pi_{q,\ep}^* 
  = \int_{\left\{
  \frac{d\pi_{q,\ep}^*}{dm} \geq 1
  \right\}}
  \left(
  \frac{d\pi_{q,\ep}^*}{dm} - 1
  \right)
  \frac{d\pi_{q,\ep}^*}{dm}
  \, dm \\
  &\leq \int_{\left\{
  \frac{d\pi_{q,\ep}^*}{dm} \geq 1
  \right\}}
  \left(
  \frac{d\pi_{q,\ep}^*}{dm}
  \right)^2
  \, dm 
  \leq \int_{X_1 \times X_2}
  \left(
  \frac{d\pi_{q,\ep}^*}{dm}
  \right)^2
  \, dm.
\end{align*}
Since $\pi_{q,\ep}^*$ is a probability measure, we have
\begin{equation*}
\int_{X_1 \times X_2}
\frac{d\pi_{q,\ep}^*}{dm}
\, dm
= 1.
\end{equation*}
Therefore, by Proposition \ref{prop:uniform-Linfty-density},
\begin{align*}
  \int_{X_1 \times X_2}
  \left( \frac{d\pi_{q,\ep}^*}{dm} \right)^2 \, dm &\leq \left\| \frac{d\pi_{q,\ep}^*}{dm} \right\|_\infty \int_{X_1 \times X_2} \frac{d\pi_{q,\ep}^*}{dm} \, dm 
  = \left\| \frac{d\pi_{q,\ep}^*}{dm} \right\|_\infty  \\
  &\leq \left(1 + \frac{5 (q - 1)}{q \ep} \|c\|_{\infty}\right)^{\frac{1}{q - 1}} 
  \leq \exp\left(\frac{5\|c\|_\infty}{\ep}\right).
\end{align*}
Thus we obtain 
\begin{equation}\label{KL;4}
D_q^+(\pi_{q,\ep}^*, m) \leq \exp\left(\frac{5\|c\|_\infty}{\ep}\right).
\end{equation}
Combining \eqref{KL;2}, \eqref{KL;3}, \eqref{KL;4}, and the estimate obtained in the previous subsection, we obtain
\begin{equation*}
\begin{split}
&\ep \inf_{\pi \in \Pi(\mu_1, \mu_2)} D_q(\pi, K_{q,\ep}) - \ep \inf_{\pi \in \Pi(\mu_1, \mu_2)} \mathrm{KL}(\pi, K_\ep) \\
&\quad \leq \left(H_\ep\ep+\|c\|_\infty \exp\left(\frac{5\|c\|_\infty}{\ep}\right)\right)(q-1),
\end{split}
\end{equation*}
where we set
\begin{equation*}
C_\ep = H_\ep\ep + \|c\|_\infty\exp\left(\frac{5\|c\|_\infty}{\ep}\right).
\end{equation*}

This proves the upper bound with a constant independent of $q$. We next prove the reverse estimate. Let $\pi=\rho m\in\Pi(\mu_1,\mu_2)$ satisfy $D_q(\pi,m)<+\infty$. Since $\mathrm{KL}(\pi,m)\leq D_q(\pi,m)$, all the terms below are finite. By
\begin{equation*}
\frac{dK_{q,\ep}}{dm}=\left(1+(q-1)\frac{c}{\ep}\right)^{-\frac{1}{q-1}}, \qquad \frac{dK_\ep}{dm}=e^{-c/\ep},
\end{equation*}
we obtain
\begin{align*}
D_q(\pi,K_{q,\ep})
&=\frac{1}{q-1}\int_{X_1\times X_2}
\left\{\rho^q\left(1+(q-1)\frac{c}{\ep}\right)-\rho\right\}\,dm \\
&=D_q(\pi,m)+\frac{1}{\ep}\int_{X_1\times X_2}c\rho^q\,dm,
\end{align*}
and
\begin{equation*}
\mathrm{KL}(\pi,K_\ep)=\int_{X_1\times X_2}\rho\log\left(\rho e^{c/\ep}\right)\,dm=\mathrm{KL}(\pi,m)+\frac{1}{\ep}\int_{X_1\times X_2}c\rho\,dm.
\end{equation*}
Therefore,
\begin{equation*}
\begin{split}
&\ep D_q(\pi,K_{q,\ep})-\ep\mathrm{KL}(\pi,K_\ep) \\
&\quad=\ep\bigl(D_q(\pi,m)-\mathrm{KL}(\pi,m)\bigr)+\int_{X_1\times X_2}c\,(\rho^q-\rho)\,dm.
\end{split}
\end{equation*}
The first term on the right-hand side is nonnegative. Furthermore,
\begin{equation*}
\min_{t\geq0}(t^q-t) =-(q-1)q^{-q/(q-1)}\geq-(q-1).
\end{equation*}
Since $0\leq c\leq\|c\|_\infty$, it follows that
\begin{equation*}
\ep D_q(\pi,K_{q,\ep}) \geq\ep\mathrm{KL}(\pi,K_\ep)-\|c\|_\infty(q-1).
\end{equation*}
Taking the infimum over $\Pi(\mu_1,\mu_2)$ gives
\begin{equation*}
\ep\inf_{\pi\in\Pi(\mu_1,\mu_2)}D_q(\pi,K_{q,\ep}) -\ep\inf_{\pi\in\Pi(\mu_1,\mu_2)}\mathrm{KL}(\pi,K_\ep) \geq-\|c\|_\infty(q-1).
\end{equation*}
Finally, $C_\ep\geq\|c\|_\infty$, and hence the upper and lower estimates yield
\begin{equation*}
\left| \ep\inf_{\pi\in\Pi(\mu_1,\mu_2)}D_q(\pi,K_{q,\ep}) -\ep\inf_{\pi\in\Pi(\mu_1,\mu_2)}\mathrm{KL}(\pi,K_\ep) \right|\leq C_\ep(q-1).
\end{equation*}

\end{pf}

\section*{Acknowledgments}
The first author is partially supported by JSPS Grant-in-Aid for Research Activity Start-up (22K20336). The second author is partially supported by JSPS Grant-in-Aid for Early-Career Scientists (21K13822) and JST PRESTO (JPMJPR24KD).


\bibliography{Tsallis}
\bibliographystyle{plain}

\bigskip

\noindent
\textsc{
Faculty of Advanced Science and Technology, Kumamoto University, Kumamoto 860-8555, Japan}\\
\noindent
\emph{Electronic mail address:}
suguro@kumamoto-u.ac.jp

\bigskip

\noindent
\textsc{ 
Mathematical Science Center for Co-creative Society, Tohoku University, 
Sendai 980-0845, Japan} \\
\noindent
\emph{Electronic mail address:}
toshiaki.yachimura.a4@tohoku.ac.jp

\end{document}